\documentclass[12pt,draftclsnofoot,letterpaper,onecolumn,romanappendices]{IEEEtran}
\usepackage[dvips]{graphicx}
\usepackage{cite}
\usepackage{epsfig}
\usepackage{amsmath,amssymb,amsfonts,amstext,amsbsy,amsopn,dsfont}
\usepackage{cases}
\usepackage{sublabel}
\usepackage{array}

\newtheorem{theorem}{\textbf{Theorem}}

\newtheorem{lemma}{\textbf{Lemma}}

\newtheorem{definition}{\textbf{Definition}}
\newtheorem{corollary}{\textbf{Corollary}}

\newcommand{\defn}{\triangleq}
\newcommand{\dif}{\textmd{d}}

\begin{document}

\title{Ergodic Transmission Capacity of Wireless Ad Hoc Networks with Interference Management}

\author{Chun-Hung Liu and Jeffrey G. Andrews
\thanks{C.-H. Liu and J. G. Andrews are with the Dept. of Electrical and Computer Engineering, the University of Texas at Austin, Austin TX 78712-0204, USA. The contact author is J. G. Andrews (Email: jandrews@ece.utexas.edu). \today.}}

\maketitle

\begin{abstract}
Most work on wireless network throughput ignores the temporal correlation inherent to wireless channels because it degrades tractability. To better model and quantify the temporal variations of wireless network throughput, this paper introduces a metric termed ergodic transmission capacity (ETC), which includes spatial and temporal ergodicity. All transmitters in the network form a homogeneous Poisson point process and all channels are modeled by a finite state Markov chain. The bounds on outage probability and ETC are characterized, and their scaling behaviors for a sparse and dense network are discussed. From these results, we show that the ETC can be characterized by the inner product of the channel-state related vector and the invariant probability vector of the Markov chain. This indicates that channel-aware opportunistic transmission does not always increase ETC. Finally, we look at outage probability with interference management from a stochastic geometry point of view. The improved bounds on outage probability and ETC due to interference management are characterized and they provide some useful insights on how to effectively manage interference in sparse and dense networks.
\end{abstract}

%\begin{keywords}
%Stochastic Geometry, network throughput, Markov Chain, Interference Management, Network Information Theory.
%\end{keywords}

%\newpage
\section{Introduction}

In the past decade, our understanding of large wireless network capacity has increased considerably, but perhaps still comprises more questions than answers, especially for realistic models. Gupta and Kumar's landmark work \cite{PGPRK00}, for example, introduced the transport capacity metric and derived scaling laws on it in a size-limited network. Another more recent example is transmission capacity proposed in \cite{SWXYJGAGDV05} which is a spatial throughput metric for Poisson-distributed transmitters in an infinite network with outage constraints. Almost all of the studies following the aforementioned approaches did not consider temporal affections. For a wireless network with long-term time-varying channels, its \emph{snapshot} throughput may not provide a full picture of how the throughput evolves over time.

In this paper, we introduce a metric capable of characterizing the network throughput induced by channels with temporal and spatial ergodicity.  This metric is called the \emph{ergodic transmission capacity} (ETC), and it measures the maximum long-term average rate (in bps/Hz) that can be sent per unit area in the network with an outage constraint. We evaluate the ETC under assuming all transmitters form a homogeneous Poisson point process (PPP) with a unique receiver. The channel models span many blocks of time, and so the throughput variations over time can be characterized with our framework but not with prior frameworks. Thus, ETC may be better able to accurately suggest how to effectively use transmissions over time and space, such as multi-antenna transmission and opportunistic scheduling, and how much such techniques will improve area spectral efficiency in a long-term sense.

%\footnote{Spatial ergodicity means that the mean of a random variable can be characterized by averaging the random variable over an infinitely large space.}

\subsection{Motivation and Related Work}
In the literature on wireless network throughput (see \cite{PGPRK00,RJLCBDSJDCHILRM01,SRKPV04,RNAR04,SWXYJGAGDV05,FXPRK06,SWJGANJ07} and the references therein), a unified and time-invariant channel model is typically adopted over the entire network, but channels in a large-scale wireless network usually are diverse across time and/or space. Using a channel model without temporal correlation does not capture how channel states evolve over time and thus the impact on network throughput from the temporal (and spatial) variations of channels is ignored, and techniques which exploit the variations and correlations cannot be properly quantified. For example, we observe that transmission techniques that increase the snapshot network throughput may not increase ETC or may even degrade it. In particular, we will show that channel-aware opportunistic transmission (CAOT), i.e. transmitting when channels are in good states, does not always improve ETC, which is perhaps surprising.

ETC requires the use of different channel models that include temporal discrepancies. We propose a finite-state Markov chain (FSMC) to model the fading channels, in particular a $m$-state Markov chain that is irreducible and positive recurrent.  Each channel undergoes path loss and fading and has an ergodic property in that its fading state has an invariant (steady-state) probability \cite{SMRLT09}. This idea can be traced back to the early work of Gilbert and Elliott \cite{ENG60}\cite{EOE63} that used a two-state Markov chain to represent \emph{good} and \emph{bad} channel conditions which was extended to a finite state case in \cite{HSWNM95}.

\subsection{Contributions}

The first contribution in this paper is the model for ETC itself. We then calculate the ETC, which requires the outage probability for each fading channel state to be found, for which we find tight closed-form bounds based on a proposed $\delta$-level interfering coverage area around a receiver. Any single transmitter in the $\delta$-level region of its unintended receiver will cause an outage if its interference power is enlarged by a factor $\delta$. We show that appropriately choosing $\delta$ admits bounds on outage probability that are much tighter than those found in previous work.

Bounds on ETC and their corresponding scaling laws for some special cases are then found. They reveal several interesting implications, e.g. the scaling of ETC for a \emph{sparse} or \emph{dense} network is
\begin{equation}\label{Eqn:ScalingLawENC}
C_{\texttt{E}}=\Theta(\Phi^{\textsf{T}}\mathbf{s}_{\epsilon}),
\end{equation}
where $C_{\texttt{E}}$ denotes ETC, $\mathbf{s}_{\epsilon}=\left[\tilde{s}^{\frac{2}{\alpha}}_1, \tilde{s}^{\frac{2}{\alpha}}_2,\cdots,\tilde{s}^{\frac{2}{\alpha}}_m\right]^{\textsf{T}}$ in which $\tilde{s}_k$ is the function of the $k$th state $s_k$ of a Markov fading channel with $m$ states, $\Phi=[\phi_1\,\,\phi_2\, \cdots\,\,\phi_m]^{\textsf{T}}\in[0,1]^{m}$ is the invariant probability vector\footnote{The physical meaning of $\phi_k$ is the fraction of time that channel state $s_k$ sojourns as time goes to infinity.} of the Markov fading channel model and $\alpha>2$ is the path loss exponent. From \eqref{Eqn:ScalingLawENC}, we notice that a single deep fading state is not necessary to have a significant negative effect on ETC if its invariant probability is very small. This point is not revealed in prior work that neglects temporal variations. Also, we observe that ETC has a geometric interpretation because it can be viewed as an inner product of two vectors. Thus, ETC is maximized when the directions of $\Phi_k$ and $\mathbf{s}_{\epsilon}$ coincide. In addition, we show that channel-aware opportunisitc transmission (CAOT), which is a scheme to allow transmitters to transmit only when their channels are in good states, does not necessarily provide an ETC gain. Although CAOT is able to increase the transmission capacity contributed by good channel states, it loses the transmission capacity contributed by bad channel states. So CAOT cannot benefit ETC if the improvement is no larger than the loss from bad channel states.

Three interference management methods -- avoidance, suppression and cancellation -- are applied to the network to reduce outage probability. Bounds on the outage probability and ETC with interference management are found which provide geometric insight into the efficiency of each technique in different scenarios. For example, we show that interference cancellation is not effective for significantly increasing ETC in a spare or dense network. Also, we show interference management can control the direction and magnitude of vector $\mathbf{s}_{\epsilon}$. Finally, we show that CAOT should not be used if interference management can significantly lower interference.

\section{Network Model and Definitions}\label{Sec:NetModDefns}

\subsection{The Network Model}
The wireless network considered in this paper is of infinite size and all nodes in the network are independently and randomly scattered.
Thus, we employ a marked homogeneous Poisson point process (PPP) $\Pi$ on the plane $\mathbb{R}^2$ to represent the locations of all transmitting nodes in the network, which can be written at time $\tau$ as\footnote{Here it is better to use $\Pi(\tau)$ instead of $\Pi$. However, to simplify notation, we will use $\Pi$ to stand for $\Pi(\tau)$ throughout this paper if ignoring time indices does not induce any ambiguity. This custom is applied to other set and variable notations.}
\begin{equation}
\Pi\defn\left\{(X_i\in\mathbb{R}^2,H_i(\tau)\in\mathbb{R}_+,e_i(\tau)): e_i(\tau)=1\right\},
\end{equation}
where $X_i$ denotes node $X_i$ as well as its location, $H_i(\tau)$ is the fading channel gain between node $X_i$ and its receiver $Y_i$, and $e_i(\tau)\in\{0,1\}$ represents the transmitting index of node $X_i$: $e_i(\tau)=1$ means node $X_i$ is transmitting; otherwise, it is idle. The intensity (density) of $\Pi$ is $\lambda$ for all $\tau\in\mathbb{N}$.

Each transmitter has a unique receiver and the distance between a TX-RX pair is a constant $d>1$. All of the transmitters use the same transmit power and the channel model between each TX-RX pair is subject to path loss and fading. So the channel gain for TX-RX pair $i$ can be written as
$$H_i(\tau)\ell(|X_i-Y_i|)=H_i(\tau)\ell(d),$$
where $H_i(\tau)$ is the fading channel gain, $|X_i-Y_i|$ denotes the Euclidean distance between nodes $X_i$ and $Y_i$ and $\ell(|\cdot|)$ is the path loss function. In order to avoid the singularity where $|X\rightarrow 0$, we will use
\begin{equation}
\ell(|X|)=|X|^{-\alpha} \mathds{1}(|X|\in[1,\infty)),\quad X\in\mathbb{R}^2,
\end{equation}
where $\alpha>2$ is the path loss exponent\footnote{In a planar network, we require $\alpha>2$ to have bounded interference, i.e. $I_t<\infty$ almost surely if $\alpha>2$ \cite{FBBBPM06,MHJGAFBODMF10}.} and $\mathds{1}(x\in\mathcal{X})$ denotes the \emph{indicator} function: $\mathds{1}(x\in\mathcal{X})=1$ if $x\in\mathcal{X}$ and 0, otherwise.

Specifically, we use an $m$-state FSMC model to characterize the fading effect of all channels in the network. The FSMC is irreducible and positive recurrent, and its $m$ states are ordered. The FSMC model with transition matrix $\mathbf{P}$ is denoted by  $\mathcal{S}(\mathbf{P})\in\mathbb{R}_+^{m}$ and $\mathcal{S}$ is an order set of the $m$ states, i.e. for any two states $s_i, s_j\in \mathcal{S}$ we have $s_i< s_j$ where $i<j$ and $i,j\in\mathcal{M}\defn\{1,2,\cdots,m\}$. Since $\mathcal{S}$ is irreducible and positive recurrent, the fading channel gain $H(\tau)$ for all TX-RX pairs must satisfy the following conditions\cite{SMRLT09}:
\begin{eqnarray}
\phi_k \defn \lim_{L\rightarrow\infty} \frac{1}{L} \sum_{\tau=0}^{L-1} \mathds{1}(H(\tau)\in s_k)\,\,\text{ and }\,\,\mathbf{u}^{\intercal}\Phi = \sum_{k=1}^m \phi_k = 1,
\end{eqnarray}
where $\Phi\defn [\phi_1,\phi_2,\cdots,\phi_m]^{\textsf{T}}$ is the invariant probability vector of $\mathbf{P}$ and $\mathbf{u}\defn[1\,\,1\,\cdots \,1]^{\textsf{T}}$ is an $m$-tuple vector. Namely, at any time $\tau$, $H(\tau)$ must belong to one of the states in $\mathcal{S}$ and $\phi_k$ represents the probability that $H(\tau)$ visits state $s_k$ in a long-term sense. We can show this  $m$-state Markov channel model has a temporal ergodic property as stated in the following lemma.
\begin{lemma}[\textbf{Temporal Ergodicity of an $m$-state Markov Chain}]\label{Lem:TemErgMarkovChain}
Suppose $\mathcal{S}(\mathbf{P})$ is an irreducible and positive recurrent Markov chain with $m$ states, and its transition matrix is $\mathbf{P}$. Let $\hbar: \mathcal{S}\rightarrow [0,1]$ be a state measurable function of $\mathcal{S}$ and $Z(\tau)$ is a random variable taking values in $\mathcal{S}$. Thus, we have
\begin{equation}
\mathbb{P}\left[\lim_{L\rightarrow\infty}\frac{1}{L}\sum_{\tau=0}^{L-1}\hbar(Z(\tau)\in\mathcal{S})=\sum_{i=1}^m \phi_k\, \hbar(s_k)\right]=1,
\end{equation}
where $\{\phi_k, k=1,\cdots,m\}$ are the invariant (steady state) distribution of $\mathbf{P}$ and $s_k$ is the $k$th state of $\mathcal{S}$.
\end{lemma}
\begin{proof}
See Appendix \ref{App:ProofTemErgMarkovChain}.
\end{proof}

The definition of ergodic transmission capacity in the following subsection is built based on the result in Lemma \ref{Lem:TemErgMarkovChain}. In addition, the following lemma shows that the fading channel model of a FSMC also has a spatial ergodic property.

\begin{lemma}[\textbf{Spatial Ergodicity}]\label{Lem:SpaErgMarkovChain}
Consider a marked homogeneous PPP $\Pi$ with an independent mark $H(\tau)\in\mathcal{S}$, and let $g: \Pi\rightarrow \mathbb{R}_+$ be a measurable function on $\Pi$. For any bounded subset $\mathcal{A}_n\subset\mathbb{R}^2$ and $\mu(\mathcal{A}_n)\rightarrow\infty$ as $n\rightarrow\infty$, we have
\begin{equation}
\mathbb{E}[g(\Pi_k)]\defn\lim_{n,\tau\rightarrow\infty}\frac{1}{\mu(\mathcal{A}_n)}\int_{\mathcal{A}_n} g((X,H\in s_k))\mu(\dif X)=\phi_k\mathbb{E}[g(\Pi)],\,\,a.s.  ,
\end{equation}
where $\Pi_k\defn\{(X_i,H_i)\in\Pi: H_i\in s_k\}$ is the PPP with channel state $s_k$.
\end{lemma}
\begin{proof}
Since $\Pi$ is homogeneous, we know
$$\mathbb{E}[g(\Pi)]=\lim_{n,\tau\rightarrow\infty}\frac{1}{\mu(\mathcal{A}_n)}\int_{\mathcal{A}_n}\mathbb{E}_{H}[g((X,H))]\mu(\dif X).$$
Since $\{H_i\}$ are independent, $\Pi_k$ is just the thinning homogeneous PPP of $\Pi$ and thus we have
\begin{eqnarray*}
\mathbb{E}[g(\Pi_k)] &=& \lim_{n,\tau\rightarrow\infty}\frac{1}{\mu(\mathcal{A}_n)}\int_{\mathcal{A}_n} g((X,H\in s_k)) \mu(\dif X)\\
  &=& \lim_{n,\tau\rightarrow\infty}\frac{1}{\mu(\mathcal{A}_n)}\int_{\mathcal{A}_n} \mathbb{E}[g((X,H))]\mathbb{E}[\mathds{1}(H\in s_k)] \mu(\dif X)\stackrel{(\star)}{=} \phi_k \mathbb{E}[g(\Pi)],
\end{eqnarray*}
where $(\star)$ follows from the temporal ergodicity result in Lemma \ref{Lem:TemErgMarkovChain}.
%See Appendix \ref{App:ProofSpaErgMarkovChain}.
\end{proof}
Lemma \ref{Lem:SpaErgMarkovChain} indicates that the spatial average of $g(\Pi_k)$ is equal to $\phi_k\mathbb{E}[g(\Pi)]$. So we know the intensity of $\Pi_k$ is $\lambda_k=\phi_k\lambda$ provided that $g(\cdot)$ is an intensity measure.

The interference channel gain from transmitter $X_j$ to its non-intended receiver $Y_i$ is denoted by $\tilde{H}_{ji}(\tau)\,\ell(|X_j-Y_i|)$ where $ \tilde{H}_{ji}(\tau)\in\mathcal{S}$. The aggregate interference normalized by the transmit power at receiver $Y_i$ can thus be expressed as
\begin{equation}\label{Eqn:Interference}
I_i(\tau) = \sum_{X_j\in\Pi\backslash X_i}  \tilde{H}_{ji}(\tau)\ell(|X_j-Y_i|),
\end{equation}
where $I_i$ is also called a spatial shot noise process \cite{SBLMCT90,JIDH98,ESSJAS90,FBBBPM06} since it captures the cumulative effect at location $Y_i$ of a set of random shocks appearing at random locations $X_j$, and $\tilde{H}_{ji}\ell(|X_j-Y_i|)$ can be viewed as the impulse function that gives the attenuation of the transmit power in space. In order to have a successful transmission for TX-RX pair $i$, the following signal-to-interference ratio (SIR) condition at receiver node $Y_i$ must hold at time $\tau$:
\begin{equation}\label{Eqn:SIR}
\mathrm{SIR}_i(\tau,\lambda) \defn \frac{H_i(\tau)}{d^{\alpha}I_i(\tau)}\geq \beta,
\end{equation}
where $\beta$ is the SIR threshold for TX-RX pair $i$ to successfully decode the received data. The network is assumed to be interference-limited.

Note that according to Slivnyak's theorem\cite{DSWKJM96} the statistics of $I_i$ seen by any node in the network is the same if the nodes form a homogeneous PPP. That means the average outage probability of each receiver node may be found by evaluating the SIR seen by a receiver located at the origin. Intuitively, the distribution of the point process is unaffected by the addition of a receiver at the origin, and this receiver is called a \emph{typical} receiver. The performance measured at the origin is often referred to the Palm measure, and in keeping with simplified notation we will denote the probability and expectation of functionals of evaluated at the origin o by $\mathbb{P}$ and $\mathbb{E}$, respectively. Also, Table \ref{Tab:MathNotation} summaries the main mathematical notation used in this paper.

\subsection{Definitions}
Consider the typical TX-RX pair and its steady state outage probability is
\begin{equation}\label{Eqn:OutageProb}
\lim_{\tau\rightarrow\infty}\mathbb{P}[\mathrm{SIR}(\tau,\lambda)<\beta] \in \{q_k(\lambda), k\in\mathcal{M} \},
\end{equation}
where $q_k(\lambda)\defn \lim_{\tau\rightarrow\infty}\mathbb{P}[\mathrm{SIR}(\tau,\lambda)<\beta|H(\tau)=s_k]$ is the outage probability for channel state $s_k$ as time $\tau$ goes to infinity. Now we are ready to use \eqref{Eqn:OutageProb} to define ergodic transmission capacity in this paper.

\begin{definition}[\textbf{Ergodic Transmission Capacity}]\label{Def:ErgTraCap}
Suppose transmitting nodes in a wireless ad hoc network form a homogeneous PPP of intensity $\lambda$. For a given $\epsilon\in(0,1)$, the ergodic transmission capacity (ETC) of a wireless ad hoc network is defined by
\begin{equation}\label{Eqn:DefTransCap}
C_{\texttt{E}} \defn b\,\bar{\lambda}_{\texttt{E}}(1-\epsilon),
\end{equation}
where $b$ is the supportable transmission rate, $\epsilon$ is the upper bound on the outage probability of each channel state, and $\bar{\lambda}_{\epsilon}=\sup\{\lambda>0: \sum_{k=1}^m\phi_k q_k(\lambda)\leq \epsilon\}$ is called maximum contention intensity achieved under the outage probability constraint $\epsilon$.
\end{definition}
The definition of $C_{\texttt{E}}$ in \eqref{Eqn:DefTransCap} originates from the following definition:
\begin{equation}\label{Eqn:OriDefEgrNetCap}
C_{\texttt{E}} \defn \frac{b\,\bar{\lambda}_{\texttt{E}}}{L}\sum_{\tau=0}^{L-1} \mathbb{P}[\mathrm{SIR}(\tau,\bar{\lambda}_{\epsilon})\geq\beta].
\end{equation}
Since all channels are an irreducible and positive current Markov chain, according to Lemma \ref{Lem:TemErgMarkovChain} they all have temporal ergodicity. Thus, the definition in \eqref{Eqn:OriDefEgrNetCap} is equivalent to
$C_{\texttt{E}} = b\,\bar{\lambda}_{\texttt{E}} \sum_{k=1}^m \phi_k[1-q_k(\bar{\lambda}_{\texttt{E}})]$.
This is the reason why we directly use the invariant probability of a Markov chain to define ETC instead of using \eqref{Eqn:OriDefEgrNetCap}. For ease of analysis, we need to quantitatively define the sparseness and denseness of a network with Poisson-distributed nodes.
\begin{definition}[\textbf{Spatial Sparseness and Denseness of a Poisson-Distributed Network}]\label{Def:NetworkStatus}
Suppose the transmission coverage of a transmitter is the circular area of radius $d$. A network whose transmitting nodes form a homogeneous PPP of intensity $\lambda$ is called ``dense'' (``sparse'') if the average number of transmitting nodes in the coverage is sufficiently large (small), i.e. $\pi d^2\lambda\gg 1$ ($\pi d^2\ll 1$).
\end{definition}

\section{General Results on Ergodic Transmission Capacity}\label{Sec:ErgNetCap}

In this section, we study the general results of ETC. First, we have to calculate the outage probability for each channel state; however, only the bounds on the outage probability and ETC can be characterized due to the complicated distribution of the interference. According to the found bounds, the scaling behaviors of ETC are characterized and they reveal several observations.

\subsection{Bounds on the Outage Probability}\label{Subsec:BoundOutProb}

Since a closed-form expression of the outage probability defined in \eqref{Eqn:OutageProb} is difficult to find\footnote{If all channels are instead Rayleigh fading, the closed-form of the outage probability can be found by the Laplace transform of the aggregate interference contributed by Poisson-distributed transmitters\cite{FBBBPM06}\cite{JFCK93}. However, such closed-form outage probability cannot be obtained for the case of channels without fading \cite{SWXYJGAGDV05}\cite{SWJGAXYGDV07} or with a single state at any time.}, we resort to bounds. The idea of approaching the lower bound for a receiver with channel state $s_k$ is to use a \emph{$\delta$-level interfering coverage} $\mathcal{I}^{\delta}_k$  for the typical receiver $Y_0$ with fading state $s_k$, and it is defined as follows\footnote{If $\delta=1$, then $\mathcal{I}^{\delta}_k$ is called the dominant interferer coverage in which a single interferer causes outage at receiver $Y_0$.}:
\begin{equation}\label{Eqn:DomIntReg}
\mathcal{I}^{\delta}_k \defn \left\{X\in \mathbb{R}^2 : \frac{s_k d^{-\alpha}}{\delta\,\ell(|X|)\tilde{H}}< \beta\right\},\,\,\delta\in[1,\infty),
\end{equation}
which means any single interferer within it can cause outage at receiver $Y_0$ with a SIR threshold $\delta\beta$. If $\mathcal{I}^{\delta}_k$ is not empty, it could contain dominant interferers and non-dominant interferers. In addition, $\Pi_k^{\delta}\defn (\Pi\cap\mathcal{I}^{\delta}_k)\setminus X_0$ is called $\delta$-level interfering point process.

The lower bound on the outage probability $q_k(\lambda)$ can be acquired by considering the outage events caused by $\Pi^{\delta}_k$. The upper bound can be approached by finding the probability of the union outage events separately caused by the interferers in $\Pi^{\delta}_k$ and $\Pi\setminus\Pi^{\delta}_k$. These two  bounds found are tighter than those in the previous works\cite{SWXYJGAGDV05}\cite{SWJGAXYGDV07}, as the following theorem shows.
\begin{theorem}\label{Thm:IneqOutProb}
The outage probability $q_k(\lambda)$ in \eqref{Eqn:OutageProb} can be bounded  as
\begin{equation}\label{Eqn:IneqOutProb}
 1- e^{-\lambda(\nu s^{-\frac{2}{\alpha}}_k-\pi)} \leq q_k(\lambda)\leq 1-\left(1-\Lambda_k(\lambda)\right)^{+} e^{-\lambda(\nu s^{-\frac{2}{\alpha}}_k-\pi)},
\end{equation}
where $(x)^+\defn \max\{x,0\}$ and $\Lambda_k(\cdot)$ is defined as
\begin{equation}\label{Eqn:Lambda_k}
\Lambda_k(\lambda)=\frac{s^{\frac{3}{\alpha}-1}_k\lambda \sigma^2}{(s^{\frac{2}{\alpha}-1}_k /d^{\alpha}\delta\beta-\lambda \eta)^2},
\end{equation}
and $\nu$, $\eta$, $\sigma^2$ are respectively given by
\begin{eqnarray*}
\nu=\pi d^2\sum_{k=1}^m \phi_k (\delta\beta s_k)^{\frac{2}{\alpha}},\,\, \eta=\frac{2\nu}{(\alpha-2)d^{\alpha}\delta\beta},\,\,\sigma^2=\frac{\pi d^{2-2\alpha}}{\alpha-1}(\delta\beta)^{\frac{1}{\alpha}-1}\sum_{k=1}^m \phi_k s_k^{\frac{1}{\alpha}+1}.
\end{eqnarray*}
\end{theorem}
\begin{proof}
See Appendix \ref{App:ProofThmIneqOutProb}.
\end{proof}

The physical meanings of $\nu$, $\eta$ and $\sigma^2$ are the mean area of $\mathcal{I}^{\delta}_k$ with $s_k=1$, the mean, and the variance of the interference contributed by the interferers  of $\Pi\setminus\Pi^{\delta}_k$ for $\lambda=s_k=1$. When the channel state $s_k$ is high, the outage probability is reduced because SIR is large or equivalently the target SIR $\beta$ is reduced. Nevertheless, it also can be explained from a geometric point of view. In \eqref{Eqn:SIR}, we can let fading gain $s_k$ be incorporated into the path loss model of all interference channels, and according to the conservation property of a homogeneous PPP \cite{DSWKJM96}, the intensity of the original PPP is changed from $\lambda$ to $\lambda/s^{\frac{2}{\alpha}}_k$. This is why $\lambda$ in the bounds is scaled by $s^{-\frac{2}{\alpha}}_k$ and thus interference generated by the PPP with intensity $\lambda/s^{\frac{2}{\alpha}}_k$ is small when $s_k$ is large. So tightness of the bounds in \eqref{Eqn:IneqOutProb} can also be observed.

If $(\cdot)^+$ in \eqref{Eqn:IneqOutProb} is non-zero, the gap between the upper and lower bounds is $\Lambda_k(\lambda) e^{-\lambda(\nu s^{-\frac{2}{\alpha}}_k-\pi)}$ which is a function of $\delta$, $\lambda$ and $s_k$. Since $e^{-\nu\lambda s^{-2/\alpha}_k}$, $\sigma^2$ and $\eta$ are all monotonically decreasing functions of $\delta$, and the denominator of $\Lambda_k(\lambda) e^{-\lambda(\nu s^{-\frac{2}{\alpha}}_k-\pi)}$ is convex for $\delta$, it is easy to realize that $\Lambda_k(\lambda) e^{-\lambda(\nu s^{-\frac{2}{\alpha}}_k-\pi)}$ is smaller than that without $\delta$ if $\delta$ is chosen appropriately. Fig. \ref{Fig:GapBounds} shows the simulation results for channel fading modeled by a 2-state Markov chain. As expected, the two gaps decrease along with $\delta$ so that using $\delta>1$ can make the bounds (much) tighter. In addition, the gap for a good channel state is much larger than that for a bad channel state. Hence, we should choose a sufficiently large $\delta$ in order to have tight bounds when the Markov chain has very good channel states.

The result in \eqref{Eqn:IneqOutProb} will become slightly different if a transmitter uses a channel-aware opportunistic transmission (CAOT) policy. Recall that the states of a Markov fading channel are ordered so that a better state has a higher subscript index. Suppose we call a channel state ``good'' in each FSMC if its subscript index is greater than or equal to $\textsf{g}$ and $s_{\textsf{g}}>1$, which means channel gain $H$ is good if $H\in\mathcal{S}_\textsf{g}\defn\{s_g,\cdots,s_m\}$. Therefore, the PPP with good channel states can be expressed as
\begin{equation}\label{Eqn:TranSetGoodState}
 \Pi_\textsf{g} = \{(X_i,H_i(\tau))\in\Pi: H_i(\tau)\in\mathcal{S}_\textsf{g}\}.
\end{equation}
According to Lemmas \ref{Lem:TemErgMarkovChain} and \ref{Lem:SpaErgMarkovChain}, its intensity is $\lambda_\textsf{g}=\lambda \mathbb{P}[H\in\mathcal{S}_\textsf{g}]=\lambda \sum_{k=\textsf{g}}^m \phi_k$ as time goes to infinity. Therefore, the bounds on the outage probability with CAOT can be obtained from \eqref{Eqn:IneqOutProb} by replacing $\lambda$ with $\lambda_\textsf{g}$, which yields
\begin{equation}\label{Eqn:IneqOutProbOppTran}
1-e^{-\lambda_\textsf{g}(\nu s^{-\frac{2}{\alpha}}_k-\pi)} \leq q_k(\lambda_\textsf{g})\leq 1-\left(1-\Lambda_k(\lambda_{\textsf{g}})\right)^+ e^{-\lambda_\textsf{g}(\nu s^{-\frac{2}{\alpha}}_k-\pi)}.
\end{equation}
Note that the bounds decrease in \eqref{Eqn:IneqOutProb} when $\lambda$ decreases. Thus, the bonds in \eqref{Eqn:IneqOutProbOppTran} decrease compared with \eqref{Eqn:IneqOutProb}. So CAOT improves the bounds on the outage probability of each channel state because nodes with bad channel states refrain from transmitting. In Section \ref{Subsec:DisErgNetCap}, however, we will point out that it may not always improve ETC.

\subsection{Ergodic Transmission Capacity}\label{Subsec:ErgTraCap}

By using the bounds on the outage probability of each channel state in Theorem \ref{Thm:IneqOutProb}, the ETC in \eqref{Eqn:DefTransCap} has bounds as shown in the following theorem.
\begin{theorem}\label{Thm:ErgTraCap}
Suppose the outage probability $q_k(\lambda)$ in \eqref{Eqn:OutageProb} is upper bounded by $\epsilon\in(0,1)$. Using the inequality in \eqref{Eqn:IneqOutProb}, bounds on the maximum contention intensity $\bar{\lambda}_{\texttt{E}}$ which maximize $\sum_{k=1}^m \phi_k q_k(\bar{\lambda}_{\texttt{E}})$ under the constraint of $\epsilon$ can be given by
\begin{equation}\label{Eqn:BoundMaxConIntKi}
\sum_{k=1}^m s_k^{\frac{2}{\alpha}}\phi_k\bar{ \lambda}^{\epsilon}_k \leq \bar{\lambda}_{\texttt{E}} \leq -\ln(1-\epsilon)\sum_{k=1}^m\frac{ s_k^{\frac{2}{\alpha}}\phi_k}{\nu- s_k^{\frac{2}{\alpha}}\pi},
\end{equation}
where $\bar{\lambda}^{\epsilon}_k$ is given by in the following.
\begin{equation}\label{Eqn:LowBoundIntenErgTranCap}
\bar{\lambda}^{\epsilon}_k = \inf\left\{\lambda>0 : \frac{1}{\lambda} \ln\left[\frac{(1-\Lambda_k(\lambda))^+}{1-\epsilon}\right] \leq \left(\nu s^{-\frac{2}{\alpha}}_k-\pi\right)\right\},
\end{equation}
where $\Lambda_k(\lambda)$ is defined in \eqref{Eqn:Lambda_k}.
\end{theorem}
\begin{proof}
See Appendix \ref{App:ProofThmErgTraCap}.
\end{proof}

\textbf{Scaling Laws of Ergodic Transmission Capacity}. If we consider $\epsilon\rightarrow 0$, the upper bound in \eqref{Eqn:BoundMaxConIntKi} will reduce to $\epsilon\sum_{k=1}^m s^{\frac{2}{\alpha}}_k\phi_k/(\nu-s^{\frac{2}{\alpha}}_k\pi)+\Theta(\epsilon^2)$ since $-\ln(1-\epsilon)\rightarrow \epsilon$. So we know $\bar{\lambda}^{\epsilon}_k\ll 1$ and thus the solution in \eqref{Eqn:LowBoundIntenErgTranCap} is $\bar{\lambda}^{\epsilon}_k=s^{\frac{2}{\alpha}}_k\epsilon/[(\sigma^2 d^{2\alpha}\beta^2\delta^2 s^{1+\frac{1}{\alpha}}_k+\nu)-s^{\frac{2}{\alpha}}_k\pi]+\Theta(\epsilon^2)$. We can conclude that the lower bounds in \eqref{Eqn:BoundMaxConIntKi} is $\epsilon\sum_{k=1}^m \phi_k s^{\frac{2}{\alpha}}_k/[(\sigma^2 d^{2\alpha}\beta^2\delta^2 s^{1+\frac{1}{\alpha}}_k+\nu)-s^{\frac{2}{\alpha}}_k\pi]+\Theta(\epsilon^2)$ as $\epsilon\rightarrow 0$. So the bounds in \eqref{Eqn:BoundMaxConIntKi} are asymptotically tight when the network is sparse because $\lambda=\Theta(\epsilon)$ as $\epsilon\rightarrow 0$. The scaling behavior of $\bar{\lambda}_{\texttt{E}}$ in this case turns out to be
\begin{equation}\label{Eqn:ScaLawErgTraCap}
\bar{\lambda}_{\texttt{E}}= \Theta\left(\frac{\epsilon}{\nu}\sum_{k=1}^m \frac{\phi_k s^{\frac{2}{\alpha}}_k}{1-\pi s^{\frac{2}{\alpha}}_k/\nu}\right),
\end{equation}
where the argument in $\Theta(\cdot)$ only keeps the \emph{key} parameters of interest.
The above result is also the scaling law for a dense network because $\nu \lambda \gtrapprox -\ln(1-\epsilon)$ as $\pi d^2\lambda$ is sufficiently large. Note that  \eqref{Eqn:ScaLawErgTraCap} is only valid for small $\lambda$ but large $\pi d^2\lambda$ since for small $\lambda$ we have to keep $\epsilon$ small and $\pi d^2$ sufficiently large in order to make $\pi d^2 \lambda$ large and the outage probability lower than $\epsilon$. In addition, if $\nu\rightarrow \infty$ (i.e. $d$ and/or $\delta\rightarrow\infty$), then \eqref{Eqn:ScaLawErgTraCap} further reduces to
\begin{equation}
\bar{\lambda}_{\texttt{E}}= \Theta\left(\frac{\epsilon}{\pi d^2 (\delta\beta)^{\frac{2}{\alpha}}}\right),
\end{equation}
which means ETC is not affected by the fading channel states if the received signal is very weak. This makes sense in that channels can be equivalently viewed in a bad state all the time when the received signal is very weak due to long transmission distance.

\textbf{ETC with Channel-Aware Opportunistic Transmission (CAOT)}.
The result in Theorem \ref{Thm:ErgTraCap} is obtained without any transmission scheduling. Suppose now the channel state information (CSI) is available at each transmitter. Then transmitters can use their CSI to do CAOT and thus we have the following corollary.
\begin{corollary}\label{Cor:ErgTraCapOppTran}
If all transmitters transmit only when their channel fading gains are in $\mathcal{S}_\textsf{g}$, the bounds on the ergodic transmission capacity are
\begin{equation}\label{Eqn:BoundErgTraCapOppTran}
\frac{1}{\varphi_{\textsf{g}}}\sum_{k=\textsf{g}}^m s^{\frac{2}{\alpha}}_k\phi_k\bar{\lambda}^{\epsilon}_k  \leq \bar{\lambda}_{\texttt{E}} \leq \frac{-\ln(1-\epsilon)}{\varphi_{\textsf{g}}}\sum_{k=\textsf{g}}^m \frac{s_k^{\frac{2}{\alpha}}\phi_k}{\nu-s_k^{\frac{2}{\alpha}}\pi},
\end{equation}
where $\varphi_{\textsf{g}}=\sum^m_{k=\textsf{g}} \phi_k$.
\end{corollary}
\begin{proof}
For CAOT, the bounds on the outage probability for channel state $s_k$ is shown in \eqref{Eqn:IneqOutProbOppTran}. Hence, the bounds in \eqref{Eqn:BoundErgTraCapOppTran} can be acquired by first taking off the terms with an index $i$ lower than $\textsf{g}$ in \eqref{Eqn:BoundMaxConIntKi} and replacing $\bar{\lambda}_{\texttt{E}}$ in the bounds with $\varphi_{\textsf{g}}\bar{\lambda}_{\texttt{E}}$. Since $\nu$, $\eta$ and $\sigma^2$ do not depend on $\lambda_\textsf{g}$, we can replace $\lambda_\textsf{g}$ in \eqref{Eqn:IneqOutProbOppTran} by $\bar{\lambda}$. Then following the same steps in the proof of Theorem \ref{Thm:ErgTraCap} to derive \eqref{Eqn:LowBoundIntenErgTranCap}, the lower bound in \eqref{Eqn:BoundErgTraCapOppTran} is completely achieved.
\end{proof}

\subsection{Observations and Discussion}\label{Subsec:DisErgNetCap}
In the previous subsection, we have obtained bounds on ETC and discussed the scaling laws of ETC for a sparse and dense network. From the bounds and scaling laws, we have made three interesting observations.

\textbf{ETC implicitly possesses a geometric interpretation}. The scaling law of ETC in \eqref{Eqn:ScaLawErgTraCap} can be expressed in a general form by vectors $\Phi$ and $\mathbf{s}_{\epsilon}$ as
\begin{equation}\label{Eqn:ErgNetCapGeoExp}
C_{\texttt{E}}=\Theta\left(\Phi^{\textsf{T}}\mathbf{s}_{\epsilon}\right),
\end{equation}
where $\Phi=[\phi_1,\phi_2,\cdots,\phi_m]^{\textsf{T}}$ and $\mathbf{s}$ is defined as
\begin{equation}
\mathbf{s}_{\epsilon}\defn \epsilon\left[\frac{s^{\frac{2}{\alpha}}_1}{\nu-\pi s^{\frac{2}{\alpha}}_1},\frac{s^{\frac{2}{\alpha}}_2}{\nu-\pi s^{\frac{2}{\alpha}}_2},\cdots,\frac{s^{\frac{2}{\alpha}}_m}{\nu-\pi s^{\frac{2}{\alpha}}_m}\right]^{\textsf{T}}.
\end{equation}
In other words, $C_{\texttt{E}}$ is scaled by the inner product of vectors $\Phi$ and $\mathbf{s}_{\epsilon}$. So the result in  \eqref{Eqn:ErgNetCapGeoExp} can be interpreted from a geometric point of view. Suppose the Markov channel model has two fading states (i.e. $m=2$, this is so called Gilbert-Elliott channel model\cite{ENG60}\cite{EOE63}). Then $C_{\texttt{E}}$ in \eqref{Eqn:ErgNetCapGeoExp} can be schematically presented in Fig. \ref{Fig:ErgNetCapGeoInt}. Note that $\mathbf{s}_{\epsilon}$ must be above on the $45^{\circ}$ line because $s_2$ is larger than $s_1$. The inner product of $\mathbf{s}_{\epsilon}$ and $\Phi$ can be written as $\Phi^{\textsf{T}}\mathbf{s}_{\epsilon}=|\Phi||\mathbf{s}_{\epsilon}|\cos\theta$ and $\theta$ is the angle between vectors $\mathbf{s}_{\epsilon}$ and $\Phi$. So we will have a larger ETC if $\Phi$ has the same direction as $\mathbf{s}_{\epsilon}$. The optimal $\Phi_*$ that maximizes ETC can be given by $\Phi_* = \frac{\mathbf{s}_{\epsilon}}{\mathbf{u}^{\textsf{T}}\mathbf{s}_{\epsilon}}$. Therefore, if all Markov fading channels have the optimal distribution $\Phi_*$, then $C_{\texttt{E}}$ in \eqref{Eqn:ErgNetCapGeoExp} becomes
\begin{equation}\label{Eqn:OptErgNetCapGeoExp}
C_{\texttt{E}}=\Theta\left(\frac{\mathbf{s}^{\textsf{T}}_{\epsilon}\mathbf{s}_{\epsilon}}{\mathbf{u}^{\textsf{T}}\mathbf{s}_{\epsilon}}\right),
\end{equation}
and thus it is completely characterized by all channel fading states.

\textbf{Dominant channel states may not dominate ETC}. Dominant states in a Markov chain means that their invariant probabilities are much larger than other states' invariant probabilities. In other words, if a channel has dominant states then it is in these states most of the time. Dominant channel states may not contribute too much ETC since their state magnitudes could be very small (very bad states). Thus, dominant channel states which really \emph{dominates} ETC only when they have a large magnitude. This point can also be visually explained by Fig. \ref{Fig:ErgNetCapGeoInt}. Suppose $s_2$ dominates and it is much larger than $s_1$. In this case, $\Phi$ and $\mathbf{s}_{\epsilon}$ will move up and be close to the vertical axis. The projection of $\Phi$ on $\mathbf{s}_{\epsilon}$ will largely increase and it is mostly contributed by the vertical component. Thus, whether a channel state dominates ETC or not depends on the product of its magnitude and invariant probability.

\textbf{CAOT may not benefit ETC}. If we compare the results in Theorem \ref{Thm:ErgTraCap} and Corollary \ref{Cor:ErgTraCapOppTran}, we can find CAOT indeed increases the bounds with good channel states. Nevertheless, \emph{it may not always improve ETC since it loses the throughput contributed by bad channel states}. To show this, let the upper bound in \eqref{Eqn:BoundMaxConIntKi} be greater than the upper bound in \eqref{Eqn:BoundErgTraCapOppTran}. This leads to the following inequality:
\begin{equation}\label{Eqn:CondOppTranNoGood}
\left(\frac{1}{\varphi_\textsf{g}}-1\right)\sum_{k=g}^m \frac{s^{\frac{2}{\alpha}}_k\phi_k}{\nu-\pi s^{\frac{2}{\alpha}}_k} < \sum_{k=1}^{g-1} \frac{s^{\frac{2}{\alpha}}_k \phi_k}{\nu-\pi s^{\frac{2}{\alpha}}_k}.
\end{equation}
The LHS in the above expression is the ETC improved by CAOT and the RHS is the ETC loss because of bad channel states. If this inequality is valid, apparently CAOT may not improve ETC because the ETC increase in the good states could not compensate the ETC loss in the bad states. Hence, following from \eqref{Eqn:CondOppTranNoGood}, the better policy of using CAOT for a transmitter is when the following condition holds
\begin{equation}
\varphi_\textsf{g} \leq \frac{\epsilon\sum_{k=g}^m s^{\frac{2}{\alpha}}_k\phi_k/(\nu-\pi s^{\frac{2}{\alpha}}_k)}{\Phi^{\textsf{T}}\mathbf{s}_{\epsilon}}.
\end{equation}
However, the above condition is not implementable by transmitters in a wireless ad hoc network without knowing $\Phi$ in advance. So CAOT is not always an effective means to enhance ETC in a real-time situation. A simulation example for ETC with and without CAOT is shown in Fig. \ref{Fig:ETCwoIM}. Channel fading is modeled by a 2-state Markov chain and the simulation condition is set to let the bad channel state be dominant. Obviously, we can see that ETC with CAOT is worse than ETC without CAOT.

\section{Ergodic Transmission Capacity with Interference Management}\label{Sec:ErgNetCapHieCod}

In Section \ref{Subsec:DisErgNetCap}, we observed that refraining from transmitting when channels are in bad states does not necessarily  increase ETC. This is because the CAOT scheme increases the transmission capacity for good channel states but further lowers the transmission capacity for users in bad channel
states. That is, the throughput increase does not in general compensate for the throughput loss, particularly when bad channel states are dominant (i.e. channels are bad most of the time). Thus, the key to increasing ETC is to boost every entry of vector $\mathbf{s}_{\epsilon}$ and not to sacrifice transmission opportunities of users with bad channel states. Two possible approaches to attaining this goal are through power control and interference management. In general, power control makes the analysis of the outage probability more intractable due to the complex structure of the interference. In addition, it is not a very effective means to increase SIR in a interference-limited ad hoc work\cite{SWJGAXYGDV07}\cite{JGASWMH07}.

\subsection{Interference Management -- A Stochastic Geometry Perspective}

Interference management can be classified into three categories: \emph{interference avoidance, suppression and cancellation}. Avoiding interference in a wireless ad hoc network is typically accomplished by using space, time or frequency orthogonality to eliminate the co-reception of strong interferers. Frequency-hopping and CSMA are prominent examples of avoiding interference in an ad hoc network. Interference suppression deploys signal processing at the transmitter and/or receiver to (linearly) suppress interference without actually cancelling it. Direct-Sequence CDMA (DS-CDMA) is a typical example of this category. In addition, receivers can try to cancel strong interference from their nearby unintended transmitters (e.g. successive interference cancellation (SIC)\cite{TC72,JGA05,DNCTPV05}). Although avoiding interference and suppressing interference are two different methods of reducing interference, we know there exists a duality property between them in a wireless network with Poisson-distributed nodes. According to the conservation property of a homogeneous PPP\cite{DSWKJM96}\footnote{The conservation property of a homogeneous PPP with intensity $\lambda$ is that the intensity will change to $\lambda/a$ if all locations of the nodes in the PPP are scaled by a constant $\sqrt{a}$.}, these two methods both reduce the original intensity of transmitters so that their effect can be demonstrated via another homogeneous PPP with a new intensity. However, interference avoidance has a better efficiency in reducing the intensity of interferers than interference suppression\cite{SWXYJGAGDV05}\cite{JGASWMH07}.

The effect of interference cancellation can also be grasped from a geometric perspective. To explain this, suppose now any receiver $Y_i$  in the network is able to cancel some interference from its nearby interferers after the interference is avoided and/or suppressed. As time goes to infinity, the \emph{interference cancellation coverage} of receiver $Y_0$ with channel state $s_k$ is defined in the following.
\begin{equation}\label{Eqn:CancelCoverage}
\mathcal{C}^{\texttt{c}}_k = \left\{ X\in\mathbb{R}^2 : \frac{\tilde{H}|X|^{-\alpha}}{s_k d^{-\alpha}+\gamma_k I_0 } \geq
 \frac{\beta}{\beta+1} \right\},
\end{equation}
where $I_0$ is the interference of receiver $Y_0$ and $\gamma_k\in(0,1)$ is called interference reduction factor for channel state $s_k$\footnote{Reduction factor $\gamma_k$ can account for the joint effect of interference avoidance and suppression. For example, if there are $M>1$ available channels for DS-CDMA with spreading gain $G>1$, then $\gamma_k=1/GM^{\frac{\alpha}{2}}$.}. Coverage $\mathcal{C}^{\texttt{c}}_k$ means that any received interference within this region can be decoded by receiver $Y_0$ with channel state $s_k$, and all transmitters in $\mathcal{C}^{\texttt{c}}_k$ have a larger received power than transmitter $X_0$ if $\mathcal{C}^{\texttt{c}}_k\cap(\Pi\setminus X_0)$ is not empty and $\beta>1$. Also, $\mathcal{C}^{\texttt{c}}_k$ is a random compact set so that its mean Lebesgue measure $\mu(\mathcal{C}^{\texttt{c}}_k)$ is finite. If each receiver can perfectly cancel all interferers in its cancellation coverage after suppressing and/or avoiding some interference, then its SIR for channel state $s_k$ is
\begin{equation}\label{Eqn:SIRwIntMag}
\mathrm{SIR}_k = \frac{s_k d^{-\alpha}}{\sum_{X_j\in\Pi_k^{\texttt{nc}}}\tilde{H}_j(\tau)|X_j|^{-\alpha}},
\end{equation}
where $\Pi_k^{\texttt{nc}}\defn \Pi\setminus(\mathcal{C}^{\texttt{c}}_k\cap\Pi\cup X_0)$ is the noncancelable part of $\Pi$ with intensity $\lambda \gamma^{\frac{2}{\alpha}}_k$. So equation \eqref{Eqn:SIRwIntMag} essentially suggests that interference management can be equivalently reflected by constructing a new PPP with a reduced intensity. The above $\mathrm{SIR}_k$ expression will be used in the following subsection to find the bounds on the outage probability. Those bounds are used to characterize the bounds on ETC with interference management.

\subsection{Bounds on Outage Probability with Interference Management}\label{Subsec:OutProbIntMag}
Bounds on the outage probability with interference management are shown in the following theorem.
\begin{theorem}\label{Thm:OutProbIntMag}
If each receiver is able to avoid and/or suppress interference, and cancel interferers in its cancellation coverage, then bounds on the outage probability for channel state $s_k$ are given by
\begin{equation}\label{Eqn:IneqOutProbIntMag}
1-e^{-\lambda_k^{\texttt{m}}(\nu s_k^{-\frac{2}{\alpha}}-\pi)}\leq q_k(\lambda) \leq 1-\left(1-\Lambda_k\left(\gamma^{\frac{2}{\alpha}}_k\lambda\right)\right)^+ e^{-\lambda_k^{\texttt{m}}(\nu s_k^{-\frac{2}{\alpha}}-\pi)},
\end{equation}
where $\lambda_k^{\texttt{m}}\defn \gamma^{\frac{2}{\alpha}}_k\lambda \left(1-\frac{\nu^{\texttt{c}}_k}{\nu}\right)^+$ is the average intensity of the transmitters in $\mathcal{I}^{\delta}_k$ with interference management and $\nu^{\texttt{c}}_k$ is the mean Lebesgue measure of $\mathcal{C}^{\texttt{c}}_k$.
\end{theorem}
\begin{proof}
See Appendix \ref{App:ProofOutProbIntMag}.
\end{proof}

There are a couple of observations that can be drawn from Theorem \ref{Thm:OutProbIntMag}. First, the upper and lower bounds in \eqref{Eqn:IneqOutProbIntMag} are smaller than those in \eqref{Eqn:IneqOutProb}. Canceling interference can be viewed as constructing a new PPP with reduced intensity $\gamma^{\frac{2}{\alpha}}_k\lambda(1-\nu^{\texttt{c}}_k/\nu)^+$ in $\mathcal{I}^{\delta}_k$, and it is more efficient to reduce interference than suppressing interference since it completely eliminates transmitters with strong interference power and thus the $1-\nu^{\texttt{c}}_k/\nu$ term does not have an exponent of $\frac{2}{\alpha}$. Thus, we can infer that imperfect interference cancellation (i.e. interference suppression) is not as efficient as perfect interference cancellation (i.e. intensity reduction) and interference avoidance since it merely decreases transmitters' interference and does not directly reduce transmitter intensity.

Second, interference cancellation is not equally useful for all networks. For example, in a dense network, canceling the strong interferences from the nearby transmitters can significantly reduce outage probability such that network throughput is substantially increased. This point can be easily verified by letting $\lambda\nu$ be sufficiently large. In this case, $q_k(\lambda)$ is close to unity if no interference is canceled. On the other hand, for sparse networks, interference cancellation may merely have a marginal reduction in outage probability. For sufficiently small $\lambda\nu$ and $\nu>\nu^{\texttt{c}}_k$, \eqref{Eqn:IneqOutProbIntMag} can be simplified as $q_k(\lambda)= \left(\frac{\gamma_k}{s_k}\right)^{\frac{2}{\alpha}}\lambda(\nu-\nu^{\texttt{c}}_k)+O((\lambda\nu)^2)$. So when receivers cancel more interferers, its outage probability is reduced by amount of $O(\lambda\nu^{\texttt{c}}_k)$ which is really a small and trivial improvement. From this observation, we see that canceling strong interferers for each channel state may not be an effective means to increase transmission capacity for small $\epsilon$ since the maximum contention intensity of transmission capacity for each channel state is already a small value in this case.

\subsection{ETC with Interference Management}\label{Subsec:ErgTraCapIntMag}
According to the bounds on the outage probability in Theorem \ref{Thm:OutProbIntMag}, ETC with interference management is bounded as shown in the following theorem.
\begin{corollary}\label{Thm:ErgTraCapIntMag}
Suppose interference management is used in the network and each receiver is able to perfectly cancel all interferers in the interference cancellation coverage of each channel state. Let the outage probability be upper bounded by $\epsilon\in(0,1)$. If $\mathcal{C}^{\texttt{c}}_k \subset \mathcal{I}^{\delta}_k$, the maximum contention intensity for channel state $s_k$ has the bounds given by
\begin{equation}\label{Eqn:BoundMaxConIntKiIntMag}
 \sum_{k=1}^m \left(\frac{s_k}{\gamma_k}\right)^{\frac{2}{\alpha}}\phi_k\bar{\lambda}^{\epsilon}_k \leq \bar{\lambda}_{\texttt{E}} \leq - \ln(1-\epsilon)\sum_{k=1}^m\left(\frac{s_k}{\gamma_k}\right)^{\frac{2}{\alpha}}\frac{\phi_k}{(\nu-s^{\frac{2}{\alpha}}_k\pi)(1-\nu^{\texttt{c}}_k/\nu)},
\end{equation}
where $\nu>\nu^{\texttt{c}}_k$ and $\bar{\lambda}^{\epsilon}_k$ is given by
\begin{equation}\label{Eqn:LowBoundMaxConIntKiIntMag}
\bar{\lambda}^{\epsilon}_k = \inf\left\{\lambda>0 : \frac{1}{\lambda}\ln\left[\frac{\left(1-\Lambda_k(\lambda)\right)^+}{1-\epsilon}\right]
\leq\left(1-\frac{\nu^{\texttt{c}}_k}{\nu}\right)\left(\nu-s^{\frac{2}{\alpha}}_k\pi\right)\right\}.
\end{equation}
However, if $ \mathcal{I}^{\delta}_k \subseteq \mathcal{C}^{\texttt{c}}_k$ , then \eqref{Eqn:BoundMaxConIntKiIntMag} becomes
\begin{equation}
\lambda_{\texttt{E}} \geq \sum_{k=1}^m \left(\frac{s_k}{\gamma_k}\right)^{\frac{2}{\alpha}}\phi_k \bar{\lambda}^{\epsilon}_k,
\end{equation}
where $\bar{\lambda}^{\epsilon}_k=\inf\left\{\lambda>0 : \left(1/d^{\alpha}\delta\beta-\lambda\eta s_k^{1-\frac{2}{\alpha}}\gamma_k^{-\frac{2}{\alpha}}\right)^2\leq s^{1-\frac{1}{\alpha}}_k\gamma_k^{\frac{2}{\alpha}}\frac{\sigma^2}{\epsilon}\lambda\right\}$.
\end{corollary}
\begin{proof}
See Appendix \ref{App:ProofThmErgTraCapIntMag}.
\end{proof}

By comparing \eqref{Eqn:BoundMaxConIntKiIntMag} with \eqref{Eqn:BoundMaxConIntKi}, we can perceive that the effect of interference cancellation on ergodic transmission capacity can be interpreted to shrink $\nu$ by $(1-\nu^{\texttt{c}}_k/\nu)$-fold. This is equivalent to saying channel gain $s_k$ increases $(1-\nu^{\texttt{c}}_k/\nu)^{-\frac{\alpha}{2}}$-fold. Suppose $\epsilon$ is small and thus from \eqref{Eqn:BoundMaxConIntKiIntMag} and \eqref{Eqn:LowBoundMaxConIntKiIntMag} we know $\bar{\lambda}_{\texttt{E}}\approx \epsilon\sum_{k=1}^m (\frac{s_k}{\gamma_k})^{\frac{2}{\alpha}}\frac{\phi_k}{(\nu-\pi s^{\frac{2}{\alpha}}_k)(1-\nu^{\texttt{c}}_k/\nu)}$. So interference cancellation in a sparse network can make the transmission capacity for channel state $s_k$ increase $(1-\nu^{\texttt{c}}_k/\nu)^{-1}$-fold. Although avoiding and suppressing interference can linearly augment ETC in a sparse network, interference cancellation could contribute much more ETC than them if $\nu^{\texttt{c}}_k$ is very close to $\nu$. Fig. \ref{Fig:ETCwIM} presents a simulation example showing how interference management improves ETC. We first notice that interference management does not provide too much ETC gain when $\epsilon$ is extremely small. On the other hand, if the network is very dense, the efficacy of interference cancellation is seriously weakened because interference is large. So the solid-circle curve of ETC looks like a concave function of $\epsilon$. Therefore, interference management in an extremely sparse or dense network can merely have marginal improvement on outage probability.

How should interference management be used for each channel state to maximize ETC? In Section \ref{Subsec:DisErgNetCap}, we have pointed out that ETC has a geometric interpretation since its bounds can be viewed as the inner product of two vectors: vectors $\Phi$ and $\mathbf{s}_{\epsilon}$ should roughly align. Since vector $\Phi$ is a channel characteristic, it cannot be manipulated to the desired direction. Therefore, the only option is to design vector $\mathbf{s}_{\epsilon}$ such that it is enlarged and rotated to the direction of $\Phi$ as closely as possible. This can be attained by interference management. To illustrate the idea of how to change $\mathbf{s}_{\epsilon}$, the right part of Fig. \ref{Fig:ErgNetCapGeoInt} is redrawn in Fig. \ref{Fig:ErgNetCapGeoIntIntMag}. Let $\mathbf{s}_{\epsilon}=\left[\tilde{s}^{\frac{2}{\alpha}}_1\,\, \tilde{s}^{\frac{2}{\alpha}}_2 \right]^{\intercal}$ where $\tilde{s}_k=s_k\epsilon^{\frac{\alpha}{2}}\left(\nu-\pi s^{\frac{2}{\alpha}}_k\right)^{-\frac{\alpha}{2}}$ and $\mathbf{s}_{\epsilon}^*=\left[(\tilde{s}^*_1)^{\frac{2}{\alpha}}\,\,(\tilde{s}^*_2)^{\frac{2}{\alpha}} \right]^{\intercal}$ represent the optimal vector that $\mathbf{s}_{\epsilon}$ can achieve by interference management. Note that $\tilde{s}_k^* = \frac{\tilde{s}_k}{\gamma_k}(1-\nu^{\texttt{c}}_k/\nu)^{-\frac{\alpha}{2}}$, $k=1,2$. Therefore, after interference is reduced, vertices $\textbf{c}$ $(\tilde{s}^{\frac{2}{\alpha}}_1,0)$ and $\textbf{a}$ $(0,\tilde{s}^{\frac{2}{\alpha}}_2)$ can be maximally pushed out to vertices \textbf{h} $((\tilde{s}^*_1)^{\frac{2}{\alpha}},0)$ and \textbf{f} $(0,(\tilde{s}^*_2)^{\frac{2}{\alpha}})$. Vertex $\textbf{i}$ is the point where $\mathbf{s}_{\epsilon}^*$ is projected on $\Phi$ and Vertex $\textbf{d}$ is the point where $\mathbf{s}_{\epsilon}$ is projected on $\Phi$. The distance from \textbf{m} to \textbf{e} represents the increase of ETC due to interference management.

%\footnote{Note that $\mathbf{s}_{\epsilon}^*$ may not be in the same direction of $\Phi$ since $\tilde{s}_1$ and $\tilde{s}_2$ may not be increased to the desired values that result in $\theta^*=0$.}

Since $\mathbf{s}_{\epsilon}^*$ is the best vector $\mathbf{s}_{\epsilon}$ can achieve, how can we make vector $\mathbf{s}_{\epsilon}$ move to vector $\mathbf{s}_{\epsilon}^*$? Namely, how should we choose $\gamma_k$ and $\nu^{\texttt{c}}_k$ for each channel state $s_k$ such that $\mathbf{s}_{\epsilon}$ can approach $\mathbf{s}_{\epsilon}^*$? The policy is to reduce interference for each channel state as much as possible because we can formulate the following nonlinear programming problem to optimize all $\gamma_k$:
\begin{eqnarray}
&\max_{\gamma_k}&\,\sum_{k=1}^m \left(\frac{\phi_k s^{\frac{2}{\alpha}}_k}{\nu-\pi s_k^{\frac{2}{\alpha}}}\right)^{\frac{\alpha}{2}} \frac{1}{\gamma_k\left[(1-\nu^{\texttt{c}}_k/\nu)\right]^{\frac{\alpha}{2}}}\\
&\text{subject to }& \gamma_k \geq \gamma_{\min_k},\quad \text{for all}\,\,k\in\mathcal{M},
\end{eqnarray}
where $\gamma_{\min_k}$ is the lower bound of $\gamma_k$ and it can be determined by the system resources or limitations such as number of available channels and the maximum spreading gain, etc. Note that $\nu^{\texttt{c}}_k$ is a monotonically decreasing and nonlinear function of $\gamma_k$ so that $1/\gamma_k\left[(1-\nu^{\texttt{c}}_k/\nu)\right]^{\frac{\alpha}{2}}$ is also a monotonically decreasing function of $\gamma_k$. Therefore, the optimal solution of $\gamma_k$ must happen at $\gamma_k=\gamma_{\min_k}$, which means interference should be avoided, suppressed and cancelled as much as possible in order to achieve $\mathbf{s}^*_{\epsilon}$. In addition, using interference management could make CAOT perform poorly because it may make most of channel states become ``good'' so that $\varphi_{\textsf{g}}\approx 1$.

\section{Conclusions}

In this paper, we presented a long-term look at the transmission capacity problem, which is completely different from the previous works on investigating network throughput at a particular time point. The motivation of this work is to understand how the temporal characteristic of a channel influences the network throughput with an outage probability constraint. Therefore, all channels are modeled by a $m$-state FSMC that has temporal and spatial ergodic properties. Bounds the on outage probability of each channel state and ETC for the case with and without interference management are all found and they show that ETC can be characterized by the inner product of vectors $\Phi$ and $\mathbf{s}_{\epsilon}$. For a sparse or dense network, the scaling law of ETC that is derived from those bounds provides some guidelines on when to use CAOT and how to do interference management.

\appendix\label{Sec:AppProofThms}
\subsection{Proof of Lemma \ref{Lem:TemErgMarkovChain}}\label{App:ProofTemErgMarkovChain}
\begin{proof}
We need to show that $\frac{1}{L}\sum_{\tau=0}^{L-1}\hbar(Z(\tau))$ converges to $\sum_{i=1}^m \phi_i\, \hbar(s_i)$ as $L\rightarrow\infty$ almost surely. Suppose $\{Z(\tau),\tau\geq 0\}$ has an invariant distribution $\{\phi_i,i=1,2,\cdots,m\}$ and define $V_i(L)\defn \sum_{\tau=0}^{L-1}\mathds{1}_{s_i}(Z(\tau))$ is the number of visits to state $s_i$ before $L$. Since $\hbar$ is a positive function, for any $\mathcal{J}\subseteq \mathcal{S}$ we have
\begin{eqnarray*}
    &&\bigg|\frac{1}{L}\sum_{\tau=0}^{L-1} \hbar(Z(\tau))-\sum_{i=1}^m \phi_i \hbar(s_i)\bigg|\leq \bigg|\sum_{s_i\in\mathcal{S}}\left(\frac{V_i(L)}{L}-\phi_i\right)\hbar(s_i) \bigg|\\
    &&\leq \sum_{s_i\in\mathcal{J}}\bigg|\frac{V_i(L)}{L}-\phi_i\bigg|\hbar(s_i)+\sum_{s_i\notin \mathcal{J}}\bigg|\frac{V_i(L)}{L}-\phi_i\bigg|\hbar(s_i)\\
    &&\leq \sum_{s_i\in\mathcal{J}}\bigg|\frac{V_i(L)}{L}-\phi_i\bigg|\hbar(s_i)+\sum_{s_i\notin \mathcal{J}}\bigg|\frac{V_i(L)}{L}+\phi_i\bigg|\hbar(s_i)
    \stackrel{(\star)}{\leq} 2\sum_{s_i\in\mathcal{J}}\bigg|\frac{V_i(L)}{L}-\phi_i\bigg|\hbar(s_i)+2\sum_{s_i\notin \mathcal{J}}\phi_i\hbar(s_i).
\end{eqnarray*}
where $(\star)$ follows from $\mathbb{P}[\lim_{L\rightarrow\infty} V_i(L)/L=\phi_i]=1$ for all $i=1,2,\cdots,m$.

For a given $\varepsilon>0$, choose $\mathcal{J}$ with an appropriate size and consider $L$ is sufficiently large so that
$$\sum_{i\notin\mathcal{J}}\phi_i\hbar(s_i)< \frac{\varepsilon}{4}\quad \text{and}\quad \sum_{s_i\in\mathcal{J}}\bigg|\frac{V_i(L)}{L}-\phi_i\bigg|\hbar(s_i)<\frac{\varepsilon}{4}.$$
Therefore, when $L$ is sufficiently large it follows that
$$\bigg|\frac{1}{L}\sum_{\tau=0}^{L-1} \hbar(Z(\tau))-\sum_{i=1}^m \phi_j \hbar(s_i)\bigg|<\varepsilon,$$
which establishes the desired convergence. The proof is complete.
\end{proof}

\subsection{Proof of Theorem \ref{Thm:IneqOutProb}} \label{App:ProofThmIneqOutProb}

First of all, we have to find the intensity $\lambda_{\delta}$ of $\Pi^{\delta}_k$. According to \cite{FBBB10}, the Laplace functional of a homogeneous PPP $\Pi$ for a nonnegative function $w : \mathbb{R}^2\rightarrow\mathbb{R}_+$ is given by
\begin{equation}\label{Eqn:LapTranPPP}
\mathcal{L}_{\Pi}(w)\defn \mathbb{E}\left[e^{-\int_{\mathbb{R}^2}w(X)\Pi(\dif X)}\right]=\exp\left(-\int_{\mathbb{R}^2}\lambda(1-e^{-w(X)})\mu(\dif X)\right).
\end{equation}
Since the Laplace functional completely characterizes the distribution of a point process, we can find
the intensity of $\Pi^{\delta}_k$ by calculating $\mathcal{L}_{\Pi^{\delta}_k}(w)$. For a bounded Borel set $\mathcal{A}\subset\mathbb{R}^2$, The Laplace functional of $\Pi^{\delta}_k$ with $w(X)=\tilde{w}(X)\mathds{1}_{\Pi^{\delta}_k}(X)$ can be written as follows:
\begin{eqnarray*}
\mathcal{L}_{\Pi^{\delta}_k}(w)&=& e^{-\lambda\mu(\mathcal{A})}\sum_{i=0}^{\infty} \frac{\lambda^i}{i!}\int_{\mathcal{A}}\cdots\int_{\mathcal{A}}\prod_{j=1}^i\left(e^{-w(X_j)}\mathbb{P}[X_j\in\Pi^{\delta}_k]+\mathbb{P}[X_j\notin\Pi^{\delta}_k]\right)\mu(\dif X_1)\cdots\mu(\dif X_i)\\
&=& e^{-\lambda \mu(\mathcal{A})}\sum_{i=0}^{\infty} \frac{1}{i!}\left\{\int_{\mathcal{A}}\left(e^{-g(Y)}\mathbb{P}[Y\in\Pi^{\delta}_k]+1-\mathbb{P}[Y\in\Pi^{\delta}_k]\right)\lambda\mu(\dif Y)\right\}^i\\
&\stackrel{(a)}{=}& \exp\left(-\lambda \int_{\mathcal{A}}\left(1-e^{-g(Y)}\right)\left[\sum_{i=1}^m \mathds{1}\left(s_i\in\left[\frac{s_k d^{-\alpha}}{\ell(|Y|)\delta\beta},\infty\right)\right)\phi_i\right]\mu(\dif Y)\right),
\end{eqnarray*}
where $(a)$ follows from the property of spatial ergodicity. Letting $\mathcal{A}\rightarrow\mathbb{R}^2$ and according to \eqref{Eqn:LapTranPPP}, we know the intensity of $\Pi^{\delta}_k$ is
\begin{equation}
\lambda^{\delta}_k(x) = \lambda \sum_{i=1}^m \phi_i\mathds{1}\left(s_i\in\left[\frac{s_k d^{-\alpha}}{\ell(x)\delta\beta},\infty\right)\right),\quad x\in\mathbb{R}_+.
\end{equation}
So $\Pi^{\delta}_k$ is a non-homogeneous PPP since $\lambda^{\delta}_k$ depends on $x$.

Since $\Pi^{\delta}_k(\mathcal{I}^{\delta}_k)$ is a Poisson random variable, its mean can be found as follows:
\begin{eqnarray*}
\mathbb{E}[\Pi^{\delta}_k(\mathcal{I}^{\delta}_k)] &=& \mathbb{E}\left[\sum_{X_j\in\Pi\setminus X_0} \mathds{1}_{\Pi^{\delta}_k}(X_j)\right]\stackrel{(b)}{=}  \int_{\mathbb{R}^2} \mathbb{E}[\mathds{1}_{\Pi^{\delta}_k}(X)] \mu(\dif X)\nonumber\\
&=& \int_{\mathbb{R}^2} \lambda_{\delta}(|X|)\, \mu(\dif X)
= 2\pi\lambda \int_0^{\infty}\left[\sum_{i=1}^m \mathds{1}\left(s_i\in\left[\frac{s_k d^{-\alpha}}{\ell(x)\delta\beta},\infty\right)\right)\phi_i\right] x\, \dif x \\
&=& 2\pi\lambda\sum_{i=1}^m \phi_i \int_1^{d\sqrt[\alpha]{\frac{\delta\beta s_i}{s_k}}}x \dif x = \pi\lambda s_k^{-\frac{2}{\alpha}}\left(d^2(\delta\beta)^{\frac{2}{\alpha}} \sum_{i=1}^m \phi_i s_i^{\frac{2}{\alpha}}-s_k^{\frac{2}{\alpha}}\right)= \lambda( s^{-\frac{2}{\alpha}}_k\nu-\pi).
\end{eqnarray*}
where $(b)$ follows from the Campbell theorem \cite{DSWKJM96}. Let $\mathcal{E}(\Pi^{\delta}_k)$ denote the outage event caused by any transmitters in $\Pi^{\delta}_k$ and its probability is
$$\mathbb{P}[\mathcal{E}(\Pi^{\delta}_k)]=1-\exp\left(-\lambda\left(\nu s^{-\frac{2}{\alpha}}_k-\pi\right)\right)\leq q_k(\lambda),$$
which is a lower bound of  $q_k(\lambda)$ because it ignores the interference contributed by the transmitters that are not in $\Pi^{\delta}_k$.

Let $\mathcal{E}^{\textsf{c}}(\Pi^{\delta}_k)$ be the complement event of $\mathcal{E}(\Pi^{\delta}_k)$ and $\mathcal{E}^{\textsf{c}}(\Pi^{\delta}_k)$ means the outage event caused by the transmitters of $\Pi\setminus\Pi^{\delta}_k$. So the upper bound of $q_k(\lambda)$ is given by
\begin{eqnarray}\label{Eqn:UppBondProof}
q_k(\lambda) \leq \mathbb{P}[\mathcal{E}(\Pi^{\delta}_k)\cup \mathcal{E}^{\textsf{c}}(\Pi^{\delta}_k)]
= 1-e^{-\lambda(\nu s^{-\frac{2}{\alpha}}_k-\pi)}+\mathbb{P}[\mathcal{E}^{\textsf{c}}(\Pi^{\delta}_k)]e^{-\lambda(\nu s^{-\frac{2}{\alpha}}_k-\pi)},
\end{eqnarray}
where $\mathbb{P}[\mathcal{E}^{\textsf{c}}(\Pi^{\delta}_k)]=\mathbb{P}[s_kd^{-\alpha}< \beta I^{\textsf{c}}_k]$ and $I^{\textsf{c}}_k$ is the interference contributed by the transmitters of  $\Pi\setminus\Pi^{\delta}_k$. Unfortunately, it is impossible to explicitly calculate $\mathbb{P}[\mathcal{E}(\mathcal{I}^{\textsf{c}}_k)]$ and thus we resort to find its upper bound by Chebyshev's inequality. Using Campbell's theorem, the mean and variance of interference $I^{\textsf{c}}_k$ can be calculated as follows:
\begin{eqnarray*}
\mathbb{E}[I^{\textsf{c}}_k] = \mathbb{E}\left[\sum_{X_l\in\Pi}\tilde{H}_l\ell(|X_l|)\mathds{1}_{\Pi\setminus\Pi^{\delta}_k}(X_l)\right]= \mathbb{E}[\tilde{H}]\int_{\mathbb{R}^2} \ell(|X|)\left[\lambda-\lambda_{\delta}(|X|)\right]\mu(\dif X)=\lambda s^{1-\frac{2}{\alpha}}_k\, \eta,
\end{eqnarray*}
\begin{eqnarray*}
\mathrm{Var}[I^{\textsf{c}}_k]= \mathbb{E}\left[(I^{\textsf{c}}_k)^2\right]-\mathbb{E}[I^{\textsf{c}}_k]^2=\mathbb{E}[\tilde{H}^2]\int_{\mathbb{R}^2} [\ell(|X|)]^2 \left[\lambda-\lambda_{\delta}(|X|)\right]\mu(\dif X)=\lambda s^{1-\frac{1}{\alpha}}_k\,\sigma^2,
\end{eqnarray*}
where $\sigma^2$ is bounded is due to bounded $\eta$ since $[\ell(x)]^2 \leq \ell(x)$.
The upper bound of $\mathbb{P}[\mathcal{E}^{\textsf{c}}(\Pi^{\delta}_k)]$ can be obtained by
\begin{eqnarray*}
  \mathbb{P}[\mathcal{E}^{\textsf{c}}(\Pi^{\delta}_k)] = \mathbb{P}[d^{-\alpha}< \delta\beta I^{\textsf{c}}_k]
   \leq \mathbb{P}\left[\frac{d^{-\alpha}-s^{1-\frac{2}{\alpha}}_k\delta\beta\lambda \eta}{|d^{-\alpha}-s^{1-\frac{2}{\alpha}}_k\delta\beta\lambda \eta|} \leq \frac{\delta\beta|I^{\textsf{c}}_k-s^{1-\frac{2}{\alpha}}_k\lambda \eta|}{|d^{-\alpha}-s^{1-\frac{2}{\alpha}}_k\delta\beta\lambda \eta|}\right]\stackrel{(c)}{\leq} \frac{s^{\frac{3}{\alpha}-1}_k\lambda \sigma^2}{(s^{\frac{2}{\alpha}-1}_k/d^{\alpha}\delta\beta-\lambda \eta)^2},
\end{eqnarray*}
where $(c)$ follows from Chebyshev's inequality. Substituting the above result into \eqref{Eqn:UppBondProof}, the proof is complete.

\subsection{Proof of Theorem \ref{Thm:ErgTraCap}} \label{App:ProofThmErgTraCap}

Since $\sum_{k=1}^m \phi_k q_k(\lambda)\leq \epsilon$, the lower bound in \eqref{Eqn:IneqOutProb} with $\bar{\lambda}_{\texttt{E}}$ can be rewritten as
\begin{equation*}
\sum_{k=1}^m \phi_k \exp\left(-\bar{\lambda}_{\texttt{E}}(\nu s^{-\frac{2}{\alpha}}_k-\pi)\right)\geq 1-\sum_{k=1}^m \phi_k q_k(\bar{\lambda}_{\texttt{E}})\geq 1-\epsilon,
\end{equation*}
which gives us
\begin{equation*}
\bar{\lambda}_{\texttt{E}}\sum_{k=1}^m \phi_k e^{-\bar{\lambda}_{\texttt{E}}(\nu s^{-\frac{2}{\alpha}}_k-\pi)}\geq \bar{\lambda}_{\texttt{E}}\Rightarrow \sum_{k=1}^m  \bar{\lambda}_k^{\epsilon}\phi_k e^{-\bar{\lambda}_k^{\epsilon}(\nu s_k^{-\frac{2}{\alpha}}-\pi)}\geq \sum_{k=1}^m \phi_k \bar{ \lambda}_k^{\epsilon}(1-\epsilon),
\end{equation*}
where $\bar{\lambda}_k^{\epsilon}=\arg\max_{\lambda}\lambda e^{-\lambda(\nu s^{-\frac{2}{\alpha}}_k-\pi)}$. Hence, $e^{-\bar{\lambda}^{\epsilon}_k(\nu s^{-\frac{2}{\alpha}}_k-\pi)}\geq 1-\epsilon$ so that we have
\begin{equation}
\bar{\lambda}^{\epsilon}_k \leq \frac{-\ln(1-\epsilon)}{s^{-\frac{2}{\alpha}}_k\nu-\pi}\Rightarrow \bar{\lambda}_{\texttt{E}}\leq \frac{-\ln(1-\epsilon)}{s^{-\frac{2}{\alpha}}_k\nu-\pi} \sum_{k=1}^m \phi_k s^{\frac{2}{\alpha}}_k .
\end{equation}
So the upper bound on $\bar{\lambda}_{\texttt{E}}$ is acquired.

Similarly, we know $1-\sum_{k=1}^m \phi_k q_k(\bar{\lambda}_{\texttt{E}})\geq 1-\epsilon$ and \eqref{Eqn:IneqOutProb} with $\bar{\lambda}_{\texttt{E}}$ can give us another lower bound on $1-\sum_{k=1}^m \phi_k q_k(\bar{\lambda}_{\texttt{E}})$. Combining these two lower bounds, it yields the following result:
\begin{equation}
1-\sum_{k=1}^m \phi_i q_i(\bar{\lambda}_{\texttt{E}})\geq \max\left\{1-\epsilon, \sum_{k=1}^m \phi_k (1-\Lambda_k\left(\bar{\lambda}_{\texttt{E}}\right))^+ e^{-\bar{\lambda}_{\texttt{E}}(\nu s^{-\frac{2}{\alpha}}_k-\pi)}\right\}.
\end{equation}
Since $(1-\Lambda_k(\bar{\lambda}_{\texttt{E}}))^+ e^{-\bar{\lambda}_{\texttt{E}}(\nu s^{-\frac{2}{\alpha}}_k-\pi)}$ is a monotonically decreasing function of $\bar{\lambda}_{\texttt{E}}$, the lower bound on $\bar{\lambda}_{\texttt{E}}$ must happen when $(1-\Lambda_k(\bar{\lambda}_{\texttt{E}}))^+ e^{-\bar{\lambda}_{\texttt{E}}(\nu s^{-\frac{2}{\alpha}}_k-\pi)}$  is equal to $1-\epsilon$, which means $\bar{\lambda}^{\epsilon}_k= \inf\{\lambda: (1-\Lambda_k(\lambda))^+ e^{-\lambda(\nu s^{-\frac{2}{\alpha}}_k-\pi)}\leq 1-\epsilon\}$. We can explicitly write the relationship $(1-\Lambda_k(\lambda))^+e^{-\lambda(\nu s^{-\frac{2}{\alpha}}_k-\pi)}\leq 1-\epsilon$ in the following.
\begin{equation*}
\ln\left[\frac{\left(1-\Lambda_k(\lambda)\right)^+}{(1-\epsilon)}\right] \leq \lambda(\nu s^{-\frac{2}{\alpha}}_k-\pi),
\end{equation*}
which implies
\begin{equation*}
\bar{\lambda}^{\epsilon}_k = \inf\left\{\lambda>0 : \frac{1}{\lambda} \ln\left[\frac{(1-\Lambda_k(\lambda))^+}{1-\epsilon}\right] \leq \left(\nu s^{-\frac{2}{\alpha}}_k-\pi\right)\right\}.
\end{equation*}
Thus, the lower bound on $\bar{\lambda}_{\texttt{E}}$ can be written as $\sum_{k=1}^m \phi_k s^{\frac{2}{\alpha}}_k\bar{\lambda}_k^{\epsilon}$. The proof is complete.

\subsection{Proof of Theorem \ref{Thm:OutProbIntMag}}\label{App:ProofOutProbIntMag}

Since the interference generated by $\Pi$ is scaled by $\gamma_k$, it is equivalent to the interference generated by a homogeneous PPP with intensity $\gamma^{\frac{2}{\alpha}}_k\lambda$. The intensity of $\Pi_k^{\texttt{nc}}$ at location $X$ can be shown to be
\begin{equation*}
\lambda_k^{\texttt{nc}}(|X|) = \gamma^{\frac{2}{\alpha}}_k\lambda\mathbb{P}\left[\frac{\tilde{H}|X|^{-\alpha}}{s_k d^{-\alpha}+ I_0}<\frac{\beta}{1+\beta}\right].
\end{equation*}
The average number of nodes in $\mathcal{C}^{\texttt{c}}_k$ is $\int^{\infty}_0 [\gamma^{\frac{2}{\alpha}}_k\lambda-\lambda_{k}^{\texttt{nc}}(|X|)]\mu(\dif X)=\gamma^{\frac{2}{\alpha}}_k\lambda\nu^{\texttt{c}}_k$. The lower bound on the outage probability can be characterized by the average number of noncancelable nodes within $\mathcal{I}^{\delta}_k$, that means we have to find $\gamma^{\frac{2}{\alpha}}_k\lambda\mu(\mathcal{C}^{\texttt{c}}_k\cap\mathcal{I}^{\delta}_k)$ which is the average number of nodes in $\Pi_k^{\texttt{nc}}\cap\mathcal{I}^{\delta}_k$.

Transmitters in $\mathcal{I}^{\delta}_k$ and $\mathcal{C}^{\texttt{c}}_k$ must satisfy the following in two equalities, respectively:
\begin{eqnarray*}
|X_j|\leq \left(\frac{\delta\beta \tilde{H}_j}{s_k}\right)^{1/\alpha}d\,\,(\text{for}\,\, X_j\in\mathcal{I}^{\delta}_k)\,\,\text{ and }
|X_j|\leq \left[\frac{\tilde{H}_j(\beta+1)}{\beta (d^{\alpha}\gamma_k I_0+s_k)}\right]^{1/\alpha} d \,\,(\text{for}\,\, X_j\in\mathcal{C}^{\texttt{c}}_k).
\end{eqnarray*}
So for any $X_j\in \mathcal{I}^{\delta}_k\cap\mathcal{C}^{\texttt{c}}_k$, we must have
\begin{equation}
|X_j|\leq d\left(\frac{\delta\beta\tilde{H}_j}{s_k}\,\min\left\{1,\frac{1+\beta}{\delta\beta^2(1+d^{\alpha}\gamma_k I_0/s_k)}\right\}\right)^{1/\alpha}.
\end{equation}
If $\mathcal{C}^{\texttt{c}}_k\subseteq \mathcal{I}^{\delta}_k$ a.s., then we must have
$$\delta \geq \frac{1+\beta}{\beta^2(1+d^{\alpha}\gamma_k I_0/s_k)},\quad a.s.\,\,\Rightarrow \delta \geq \frac{1+\beta}{\beta^2}.$$
In this case, the average number of noncancelable $\delta$-level interferers is $\lambda\gamma^{\frac{2}{\alpha}}_k (\mu(\mathcal{I}^{\delta}_k)-\mu(\mathcal{C}^{\texttt{c}}_k))$ which is equal to $\lambda \gamma_k^{\frac{2}{\alpha}} \nu(1-\nu^{\texttt{c}}_k/\nu)$. On the other hand, if $\mathcal{I}^{\delta}_k\subset \mathcal{C}^{\texttt{c}}_k$ a.s., $\delta< \frac{1+\beta}{\beta^2}$ and thus $\nu^{\texttt{c}}_k>\nu$ in this case and thus all interferes in $\mathcal{I}^{\delta}_k$ are cancelable. Combining these two cases, the average number of noncancelable $\delta$-level interferers should be written as $\nu\lambda^{\texttt{m}}_k$. Using $\lambda^{\texttt{m}}_k$ to replace $\lambda$ of the lower bound in \eqref{Eqn:IneqOutProb}, we can have the lower bound on the outage probability.

The $(\cdot)^+$ term of the upper bound is due to the outage caused by the transmitters out of $\mathcal{I}^{\delta}_k$ so that the intensity of $\Pi\setminus \Pi^{\delta}_k$ is $\gamma^{2/\alpha}_k \lambda$ because no interferers are canceled in $\Pi\setminus \Pi^{\delta}_k$. So we can just replace $\lambda$ in the $(\cdot)^+$ term of \eqref{Eqn:IneqOutProb} by $\gamma^{2/\alpha}_k\lambda$ and also replace $\lambda$ in the exponential term with $\lambda_k^{\texttt{m}}$. Then the upper bound is obtained.

\subsection{Proof of Corollary \ref{Thm:ErgTraCapIntMag}}\label{App:ProofThmErgTraCapIntMag}

By considering the given condition $\sum_{k=1}^m \phi_k q_k(\bar{\lambda}_{\texttt{E}})\leq\epsilon$, the success probability for channel state $s_k$ obtained from \eqref{Eqn:IneqOutProbIntMag} is bounded as follows.
\begin{equation*}
1-\sum_{k=1}^m \phi_k\epsilon\leq 1-\sum_{k=1}^m \phi_k q_k(\bar{\lambda}_{\texttt{E}})\leq \sum_{k=1}^m \phi_k e^{-\nu_{\delta}s^{-2/\alpha}_k\lambda^{\texttt{m}}_k}.
\end{equation*}
If $\mathcal{C}^{\texttt{c}}_k \subset \mathcal{I}^{\delta}_k$, then using the above inequality and following the same steps of finding the upper and lower bounds in the proof of Theorem \ref{Thm:ErgTraCap}, we can show the results in \eqref{Eqn:BoundMaxConIntKiIntMag} and \eqref{Eqn:LowBoundMaxConIntKiIntMag}. If $\mathcal{I}^{\delta}_k \subseteq \mathcal{C}^{\texttt{c}}_k$, then the upper bound on success probability is no longer available since all $\delta$-level interferers are canceled and thus $(\nu-\nu^{\texttt{c}}_k)^+ = 0$. As shown in the proof of Theorem \ref{Thm:ErgTraCap}, there are two lower bounds on the success probability: one is $1-\epsilon$, the other is $\sum_{k=1}^m \phi_k \left(1-\Lambda_k\left(\bar{\lambda}_k\gamma_k^{\frac{2}{\alpha}}\right)\right)^+$. Considering $\Lambda_k(\cdot)<1$, the lower bound on the success probability can be expressed as
\begin{equation*}
1-\sum_{k=1}^m \phi_k q_k(\bar{\lambda}_{\texttt{E}})\geq 1-\min\left\{\epsilon,\sum_{k=1}^m \phi_k\Lambda_k\left(\bar{\lambda}_{k}\gamma_k^{\frac{2}{\alpha}}\right)\right\},
\end{equation*}
which yields the following inequality
\begin{equation*}
\frac{s^{\frac{3}{\alpha}-1}_k\lambda \sigma^2/\gamma_k^{\frac{2}{\alpha}}}{\left[s^{\frac{2}{\alpha}-1}_k/\gamma_k^{\frac{2}{\alpha}}d^{\alpha}\delta\beta-\lambda\eta\right]^2}\geq \epsilon.
\end{equation*}
This leads to the following condition:
$\bar{\lambda}^{\epsilon}_k = \inf\left\{\lambda>0 : s^{\frac{1}{\alpha}-1}_k\gamma_k^{-\frac{2}{\alpha}}\left(1/d^{\alpha}\delta\beta-\lambda\eta s_k^{1-\frac{2}{\alpha}}\gamma_k^{-\frac{2}{\alpha}}\right)^2\leq \lambda\frac{\sigma^2}{\epsilon}\right\}$, which renders us the lower bound $\sum_{k=1}^m (s_k/\gamma_k)^{\frac{2}{\alpha}}\phi_k \bar{\lambda}^{\epsilon}_k$. This completes the proof.

% section of references
\bibliographystyle{ieeetran}
\bibliography{IEEEabrv,Ref_ErgodicTC}

%%%%%%%%%%%%%%%%%%%%%%%%%%%%%%%%%%%%%%%%%%%%%%%%%%%%
% section of figures
%%%%%%%%%%%%%%%%%%%%%%%%%%%%%%%%%%%%%%%%%%%%%%%%%

\begin{table}[!h]
  \centering
  \caption{Summary of Main Mathematical Notation}\label{Tab:MathNotation}
  \begin{tabular}{|c|c|}
  \hline
  Symbol & Definition\\ \hline
  $\Pi$  & Homogeneous PPP of transmitters \\
  $\lambda$ & Intensity (density) of $\Pi$\\
  $C_{\texttt{E}}$  & Ergodic transmission capacity\\
  $\bar{\lambda}_{\texttt{E}}$ & Maximum Contention Intensity\\
  $\epsilon$ & Upper bound of outage probability\\
  $d$ & Transmission distance of a TX-RX pair\\
  $\mathcal{S}\in\mathbb{R}^{m}_+$ & $m$-state Markov chain for modeling fading\\
  $s_k$ & $k$th state of Markov chain $\mathcal{S}$ \\
  $\phi_k$ & Invariant (steady state) probability of channel state $s_k$\\
  $H(\tau)$ & Fading channel gain at time $\tau$,\,$H(\tau)\in\mathcal{S}$\\
  $\alpha>2$ & Path loss exponent\\
  $\beta$ & SIR threshold for successful decoding\\
  $\mathcal{I}^{\delta}_k$ & $\delta$-level interfering coverage for channel state $s_k$\\
  $\delta\geq 1$ & Parameter of defining $\mathcal{I}^{\delta}_k$\\
  $\nu$ & Mean area of $\mathcal{I}^{\delta}_k$ for $s_k=1$\\
  $\mathcal{C}^{\texttt{c}}_k$ & Interference cancellation coverage for channel state $s_k$ \\
  $\gamma_k\in(0,1)$ & Interference reduction factor for channel state $s_k$\\
  $\ell(|\cdot|)$ & Path loss function\\
  $q_k(\cdot)$ & Outage probability for channel state $s_k$ \\
  $\mu(\mathcal{A})$ & Lebesgue measure of set $\mathcal{A}$\\
  \hline
  \end{tabular}
\end{table}

\begin{figure}[!h]
\centering
  \includegraphics[scale=0.6]{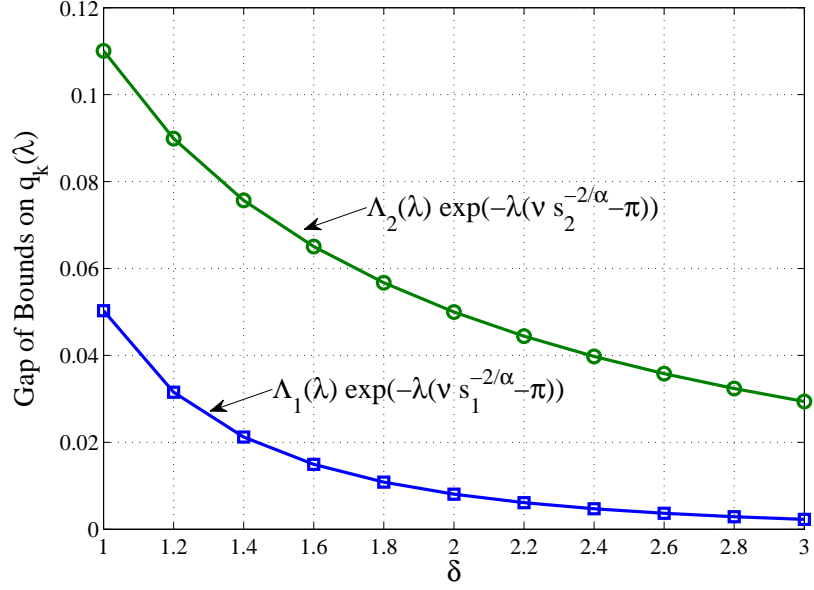}\\
  \caption{The gap between the upper and lower bounds on $q_k(\lambda)$. Channel fading is modeled by a Markov chain with 2 states. The network parameters for simulation are: $d=5m$, $\lambda=0.01$, $\alpha=3$, $\beta=2$, $s_1=0.5$, $s_2=2$ and $\phi_1=\phi_2=0.5$.}\label{Fig:GapBounds}
\end{figure}

\begin{figure}[!h]
\centering
  \includegraphics[scale=0.65]{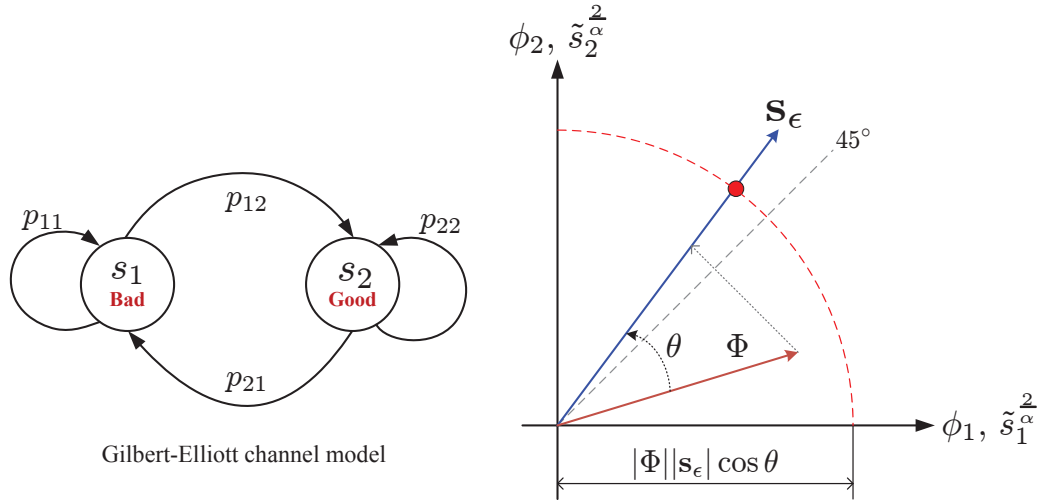}\\
  \caption{The Gilbert-Elliott channel model and its corresponding geometric presentation of $C_{\texttt{E}}$, where $\{p_{ij}\}$ are the state transition probabilities for the FSMC model and $\tilde{s}_k=s_k/(\nu-\pi s^{\frac{2}{\alpha}}_k)^{\frac{\alpha}{2}}$.}\label{Fig:ErgNetCapGeoInt}
\end{figure}

\begin{figure}[!h]
\centering
  \includegraphics[scale=0.6]{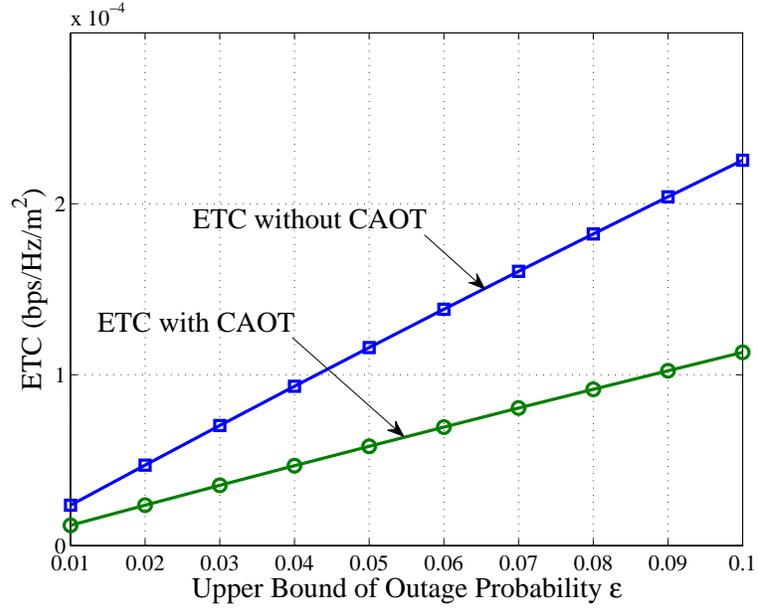}\\
  \caption{Numerical results for ETC with and without CAOT. The network parameters for simulation are : $d=10m$, $\beta=2$, $\alpha=3$, $\delta=1.5$, $s_1=0.5$, $s_2=2$, $\phi_1=0.8$, $\phi_2=0.2$ and $\lambda=0.01$.}\label{Fig:ETCwoIM}
\end{figure}

\begin{figure}[!h]
\centering
  \includegraphics[scale=0.6]{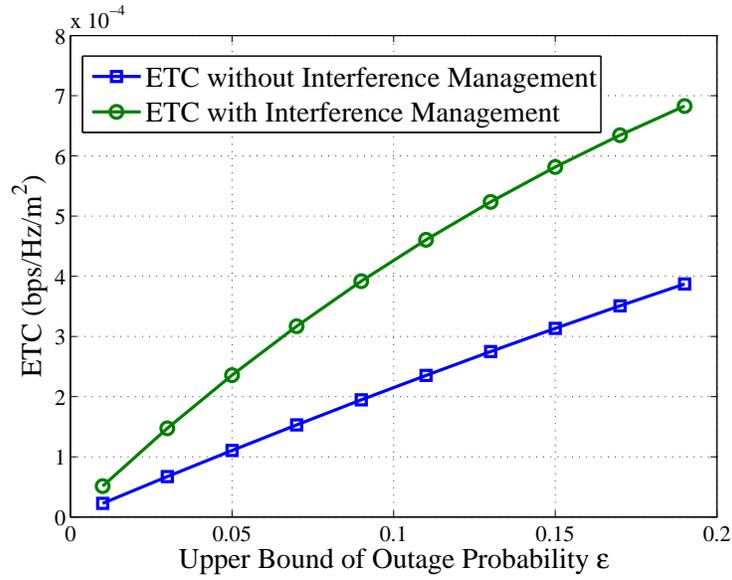}\\
  \caption{Numerical results for ETC with and without interference management. The network parameters for simulation are : $d=10m$, $\beta=2$, $\alpha=3$, $\delta=2$, $s_1=0.5$, $s_2=2$, $\phi_1=\phi_2=0.5$, $\gamma_1=\gamma_2=0.6$ and $\lambda=0.02$.}\label{Fig:ETCwIM}
\end{figure}

\begin{figure}[!h]
\centering
  \includegraphics[scale=0.7]{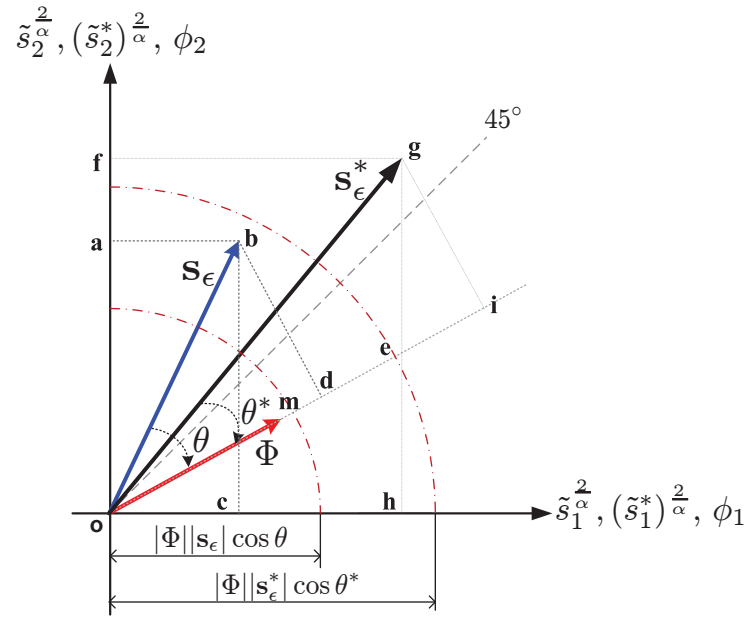}\\
  \caption{The geometric representation of the ETC for an FSMC with 2 states. $\mathbf{s}_{\epsilon}^*$ is the optimal vector that $\mathbf{s}_{\epsilon}$ can achieve by interference management. By using interference management, vertices \textbf{a} and \textbf{c} can be optimally moved to \textbf{f}  and \textbf{h}, respectively. The projection points of vertices \textbf{b} and \textbf{g} on $\Phi$ are \textbf{d} and \textbf{i}, respectively.}\label{Fig:ErgNetCapGeoIntIntMag}
\end{figure}

\end{document}